\documentclass[journal,twoside,web]{ieeecolor}
\usepackage{generic}
\usepackage{cite}
\usepackage{amsmath,amssymb,amsfonts}
\usepackage{algorithmic}
\usepackage{graphicx}
\usepackage{textcomp}
\usepackage{graphicx}
\usepackage{multicol}
\usepackage{stfloats}
\usepackage{listings}
\usepackage{booktabs}
\usepackage{makecell}
\newtheorem{theorem}{Theorem}
\newtheorem{proposition}{Proposition}
\newtheorem{corollary}{Corollary}
\newtheorem{lemma}{Lemma}
\newtheorem{definition}{Definition}
\newtheorem{assumption}{Assumption}
\newtheorem{remark}{Remark}

\newtheorem{property}{Property}
\usepackage{tikz,xcolor,hyperref}
\definecolor{lime}{HTML}{A6CE39}
\DeclareRobustCommand{\orcidicon}{%
	\begin{tikzpicture}
		\draw[lime, fill=lime] (0,0) 
		circle [radius=0.16] 
		node[white] {{\fontfamily{qag}\selectfont \tiny ID}}; 
		\draw[white, fill=white] (-0.0625,0.095) 
		circle [radius=0.007];	  
	\end{tikzpicture}
	\hspace{-2mm}}
\foreach \x in {A, ..., Z}{%
	\expandafter\xdef\csname orcid\x\endcsname{\noexpand\href{https://orcid.org/\csname orcidauthor\x\endcsname}{\noexpand\orcidicon}}
}
\makeatother

\begin{document}
\title{Distributed Economic Dispatch Scheme of A BESS Network with A Switching Utility Grid in Energy Internet Paradigm via Reset Consensus Method}
\author{Yalin Zhang\orcidA{}, Zhongxin Liu\orcidB{}, \IEEEmembership{Member, IEEE}, Yulin Chen\orcidC{}, \IEEEmembership{Member, IEEE}, Donglian Qi\orcidD{}, \IEEEmembership{Senior Member, IEEE}, and Zengqiang Chen\orcidE{}.
\thanks{This work is supported by the National Natural Science Foundation of China (Grant No. 92367105) and the General Terminal IC Interdisciplinary Science Center of Nankai University.}
\thanks{Yalin Zhang and Donglian Qi are with the College of Electrical Engineering, Zhejiang University, Hangzhou 310058, China (e-mail: zhangyalin@zju.edu.cn; qidl@zju.edu.cn).}
\thanks{Zhongxin Liu and Zengqiang Chen are with the College of Artificial Intelligence, Nankai University, Tianjin 300350, and also with the Tianjin Key Laboratory of Interventional Brain-Computer Interface and Intelligent Rehabilitation, Nankai University, Tianjin 300350, China (e-mail: lzhx@nankai.edu.cn; chenzq@nankai.edu.cn).}
\thanks{Yulin Chen is with the Hainan Institute, Zhejiang University, Sanya 572000, China (e-mail: chenyl2017@zju.edu.cn).}
}

\maketitle

\begin{abstract}
This article aims to design a distributed economic dispatch (ED) scheme based on marginal cost (MC) consensus for a battery energy storage system (BESS) network considering the utility grid switching. Due to the low control accuracy and slow convergence speed of previous schemes, some important changes have been made in this article. Firstly, the feedback coupling term is introduced into the closed-loop error to further improve the stability and convergence rate of the controller. Then, a proportional integral (PI) controller with a reset mechanism is implemented in the designed scheme, which improves the speed and control accuracy of MC consensus. The performance of a control system with a reset mechanism is well guaranteed by analyzing the eigenvalues of the corresponding matrix. Finally, some tests focus on UG switching, time phased pricing, load switching, a BESS/agent plug and play, and large-scale system applications to illustrate the effectiveness and progressiveness of the design scheme.
\end{abstract}


\section{Introduction}\label{sec1}
\subsection{Background}
The battery energy storage system (BESS) has been widely integrated as an indispensable device on every bus in the smart grid \cite{wangApplicationEnergyStorage2022}, as shown in Fig. \ref{figei}. It is used to stabilize bus voltage \cite{caleroReviewModelingApplications2022, yangModellingOptimalEnergy2022}, peak shaving and valley filling \cite{wangApplicationEnergyStorage2022, caleroReviewModelingApplications2022}, and serve as an emergency backup power source \cite{yangModellingOptimalEnergy2022, hossainlipuReviewControllersOptimizations2022}. However, the increasing operational costs caused by capacity degradation and internal consumption \cite{caleroReviewModelingApplications2022, forero-quinteroProfitabilityAnalysisDemandside2022, rouholaminiReviewModelingManagement2022} have forced us to find ways to reduce expenses.
\begin{figure}[h]
	\centering
	\includegraphics[width=8cm]{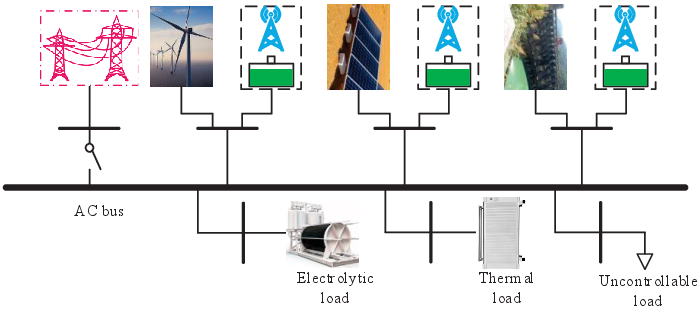} \caption{BESSs in microgrids\label{figei}}
\end{figure}
\par For a power system, designing an efficient economic dispatch (ED) algorithm is the main means of reducing operating expenses. Therefore, some centralized and distributed solutions \cite{caleroReviewModelingApplications2022, rouholaminiReviewModelingManagement2022} have been proposed one after another. Due to its good scalability and robustness, distributed solutions have attracted much attention. In view of this, we plan to spare no effort in designing an efficient distributed ED scheme for a BESS network in this article.
\subsection{Literature Review and Motivation}
\par At present, most of the existing distributed ED schemes are designed based on multi-agent systems (MASs), which need to be combined with analytical methods such as Lagrange multiplier method \cite{zhangConvergenceAnalysisIncremental2012}, gradient search method \cite{Yi2016}, Newton method \cite{andersonDistributedApproximateNewton2019}, etc. Correspondingly, a distributed ED scheme based on consensus \cite{zhangConvergenceAnalysisIncremental2012}, gradient descent \cite{Yi2016}, and approximate Newton's method \cite{andersonDistributedApproximateNewton2019}. Therein, distributed ED schemes based on marginal cost (MC) consensus \cite{tanExtensionsLocationalMarginal2022, spanglerPowerGenerationOperation2014}, i.e., a principle of equal incremental rate, have the characteristics of clear structure and easy modeling. Based on this, a series of excellent results with discrete time dynamics have been developed.
\par Initially, a distributed ED method based on MC consensus is developed in \cite{zhangConvergenceAnalysisIncremental2012} to investigate the impact of communication networks on the results. Subsequently, in \cite{6717171}, the line loss power is modeled as a constant value and considered in this method. These two basic forms continue to be developed, with some of them mainly focus on:
\par \textit{1). On communication network}. Some researchers proposed an average power mismatch estimator based on a gossip algorithm in \cite{zhangRobustDistributedSystem2016} to address the problem of packet loss in communication networks. A discrete-time second-order system based on MC consensus control is designed in \cite{yangDistributedConsensusBased2017} to ensure supply-demand balance, and the impact of communication networks on the results is investigated. Based on virtual power plants, some researchers proposed a distributed ED scheme that is robust to network attacks in \cite{liRobustDistributedEconomic2018}. Similarly, a local event driven solution capable of resisting network attacks was developed in \cite{sahooLocalizedEventDrivenResilient2021a}. To alleviate the demand for communication network balance in economic dispatch, a push based approach is proposed in \cite{wangPushBasedDistributedEconomic2021}, which is robust to gradient noise and communication delay. In addition, some researchers have established a buffer node in \cite{liConsensusBasedEnergyManagement2023} to accumulate and exchange power mismatch information in response to the problem of random packet loss in distributed ED schemes based on MC consensus.
\par \textit{2). On privacy protection}. A privacy protection mechanism based on random communication weights with error compensation mechanism is proposed in \cite{yeRandomWeightPrivacyPreservingAlgorithm2021}. In \cite{wangDisEHPPCEnablingHeterogeneous2022a}, a distributed ED scheme with dual layer privacy protection is proposed to protect the information of users and generators, where the former adopts a differential privacy protection mechanism. In order to save communication resources and improve data privacy in distributed ED schemes, the authors proposed a differential privacy protection scheme with event triggering mechanism in \cite{Yan2022a}. 
\par Through the above analysis, it is not difficult to find that existing solutions focus on the switching \cite{zhangConvergenceAnalysisIncremental2012, yangDistributedConsensusBased2017}, balance \cite{wangPushBasedDistributedEconomic2021}, and attack of communication networks \cite{liRobustDistributedEconomic2018, wangPushBasedDistributedEconomic2021}, as well as packet loss issues in the network \cite{zhangRobustDistributedSystem2016, liConsensusBasedEnergyManagement2023}, and the confidentiality of electricity and power generation data \cite{yeRandomWeightPrivacyPreservingAlgorithm2021, wangDisEHPPCEnablingHeterogeneous2022a, Yan2022a}. Without exception, due to the use of proportional (P) controllers in the designed distributed ED scheme, each variable has a low control system and a full convergence rate. At present, there are few reports on methods to improve the convergence rate of discrete-time ED schemes compared to distributed finite-/fixed-/prescribed-time \cite{Liu2022b, Dai2021, Li2022m, Ji2023} continuous time ED schemes, which is one of the motivations in this article.
\par New power systems often utilize IoT technology to organically integrate power generation/storage equipment, power consumption units, and communication equipment in a region, forming an emerging information physical system mediated by control algorithms, i.e., the Energy Internet (EI) \cite{huangFutureRenewableElectric2011}. This paper considers the ED problem of the power network composed of BESSs and a utility grid (UG) in the EI, where a UG is regulated by an energy router (ER) \cite{guoHierarchicalOptimizationStrategy2019}. For a ED problem of this system, the authors provide a basic framework in \cite{huangDistributedOptimalComultimicrogrids2016}, which is then extended to systems with multiple energy sources \cite{zhangDistributedOptimalEnergy2017, zhangDistributedAdaptiveDual2021}. In recent years, the authors have constructed a discrete-time distributed ED scheme for EI with transmission loss and UG switching in an excellent work \cite{chenDistributedEconomicDispatch2021}, where an average power mismatch estimation term with time-varying gain is introduced into the MC consensus controller. However, from the subsequent simulation results, it can be seen that when the switch of a UG is closed, applying this scheme makes it easy for each variable to oscillate, while after the switch is opened, the convergence rate of each variable is slower. Some changes must be made to address these shortcomings, which is another motivation for us to do this work.
\subsection{State of Contribution}
\par In this article, we focus on designing a unified distributed ED scheme with a reset mechanism for a BESS network to cope with the switching of a UG. From the perspective of algorithm novelty, our work is to improve the scheme in \cite{chenDistributedEconomicDispatch2021}. We design a distributed ED scheme with PI+Rest controller and further introduce the feedback term of the average power mismatch estimator in the MC control scheme into the closed-loop error. All these efforts are aimed at improving the convergence rate and control accuracy of each state. The specific contributions are as follows:
\begin{enumerate}
	\item A distributed ED scheme considering UG switching is constructed. Unlike before \cite{chenDistributedEconomicDispatch2021}, the feedback coupling term is introduced into the feedback error of the MC controller to improve the convergence rate.
	\item The PI+Rest controller is introduced to further improve the performance of the designed scheme, both in terms of convergence rate and control accuracy. In the designed scheme, the integral term does not need to be communicated.
	\item The performance of the reset mechanism has been well analyzed. The stability of UG output and the balance of power supply and demand in the network are well guaranteed.
\end{enumerate}
\subsection{The Organization of This Article}
\par Below is a brief explanation of the remaining parts of this article.
\par The configuration of the power network, the management cost of a BESS, and the optimal power flow of the power network are analyzed in section \ref{II}. 
\par The designed scheme is designed and analyzed in section \ref{III}, which includes stability analysis of the MC consensus controller and the average power estimator, conditions under which the reset mechanism acts, analysis of supply and demand balance, and stability analysis of other auxiliary variables.
\par Some simulation cases are designed to test the designed ED scheme in section \ref{IV}, including time-based electricity prices, load switching, plug and play characteristics, and a wide area system.
\par Finally, a conclusion on this paper is drawn in section \ref{V}.

\section{Preliminaries}\label{II}
	\subsection{The Optimal ED for A BESS Network with A Switching UG}
	\par With the analysis in \cite{10488096, Zhang02092025}, the cost function of a BESS network with a switching UG is modeled as
	$$\begin{aligned}
		&\mathrm{min}\;F=\sum_{i=1}^{n}f_i+gf_{UG},\\
		&\mathrm{s.t.}\;\sum_{i=1}^{n}P_i+gP_{UG}=\sum_{i=1}^{n}(P_{D,i}+P_{L,i}), P_{m,i}\le P_i\le P_{M,i},
	\end{aligned}$$
	where $g$ is a UG switching signal, $g=1$ implies that a BESS network is in the grid connected mode, otherwise it is in the isolated mode.
	\par When $g=1$, a Lagrangian function is constructed for a grid-connected BESS network as follows
	$$\begin{aligned}
		&L_a(P_1,\cdots,P_n,P_{UG})\\
		=&\sum_{i=1}^{n}f_i+f_{UG}+\omega_1(\sum_{i=1}^{n}(P_{D,i}+P_{L,i})
		-\sum_{i=1}^{n}P_i-P_{UG})\\
		&+\sum_{i=1}^{n}\omega_{2,i}(P_i-P_{M,i})
		+\sum_{i=1}^{n}\omega_{3,i}(P_{m,i}-P_i),
	\end{aligned}$$
	where $\omega_1$, $\omega_{2,i}$, and $\omega_{3,i}$ for $i\in\{1,2,\cdots,n\}$ are Lagrangian multipliers for BESS $i$, respectively. By calculating the gradient of the above, one can get
	$$\frac{\partial L_a(P_1,\cdots,P_n,P_{UG})}{\partial P_i}=2\beta_iP_i+\alpha_i+\omega_1(2Z_i P_i-1),$$
	$$\frac{\partial L_a(P_1,\cdots,P_n,P_{UG})}{\partial P_{UG}}=\rho-\omega_1.$$
	Thus, the best working point for BESS $i$ with $\forall i\in\{1,2,\cdots,n\}$ is
	\begin{equation}
		\label{gsmmc}
		\left\{\begin{array}{l}
			\omega_1^*=\frac{2\beta_iP_i+\alpha_i}{1-2Z_i P_i},\\
			\omega_1^*=\rho.
		\end{array}\right.
	\end{equation}
	\par When $g=0$, through similar calculations, the optimal operating point of an isolated BESS network can be obtained, i.e, $\forall i\in\{1,2,\cdots,n\}$,
	\begin{equation}
		\label{gsmmc1}
		\omega_1^*=\frac{2\beta_iP_i+\alpha_i}{1-2Z_i P_i}.
	\end{equation}
	\par Given the above analysis and current configuration, an optimal power output rule is constructed as follows
	\begin{equation}
		\label{power}
		P_i^*=\left\{\begin{aligned}
			P_{m,i},if\;\frac{\omega_1^*-\alpha_i}{2(\beta_i+Z_i\omega_1^*)}<P_{m,i},\\
			\frac{\omega_1^*-\alpha_i}{2(\beta_i+Z_i\omega_1^*)},if\;P_{M,i}\le\frac{\omega_1^*-\alpha_i}{2(\beta_i+Z_i\omega_1^*)}\le P_{m,i},\\
			P_{M,i},if\;\frac{\omega_1^*-\alpha_i}{2(\beta_i+Z_i\omega_1^*)}>P_{M,i}.
		\end{aligned}
		\right.
	\end{equation}
	\par When $g=1$, the optimal exchange power $P_{UG}$ is
	\begin{equation}
		\label{powermg}
		P_{UG}^*=\sum_{i=1}^{n}(P_{D,i}+P_{L,i}^*-P_i^*),
	\end{equation}
	where $P_{L,i}^*=Z_i(P_i^*)^2$, and $\sum_{i=1}^{n} P_{L,i}^*$ is so-called the optimal line loss power. When $g=0$, $\sum_{i=1}^{n}(P_{D,i}+P_{L,i}^*-P_i^*)=0$ is required.
	\begin{remark}
		Solving the ED problem of a BESS network can obtain the optimal power flow, which provides a power set point for each BESS dual closed-loop controller. Under the drive of a dual closed-loop controller, each inverter can output power according to the solution of ED.
	\end{remark}
	\subsection{Control Objectives}
	\par For a BESS network considering UG switching, it is necessary to meet the plug and play requirements of a UG while designing ED schemes specifically for grid-connected and isolated modes. This requires improving the existing MC consensus control scheme and routing scheme. Therefore, the control objectives of this article are summarized in detail as follows
	\begin{itemize}
		\item Define pseudo MC, abbreviated as MC, for BESS $i$ as $\lambda_i=\frac{2\beta_iP_i+\alpha_i}{1-2Z_i P_i}$.
		Thus, MCs can achieve consensus when $g=0$, i.e., 
		$$\lim\limits_{k\to+\infty}\vert\lambda_i-\lambda_j\vert=0,\;\forall i,j\in\{1,2,\cdots,n\}.$$
		Further, MCs converge asymptotically to the real-time electric price $\rho$ of a UG when $g=1$, i.e.,
		$$\lim\limits_{k\to+\infty}\lambda_i=\rho,\;\forall i\in\{1,2,\cdots,n\},$$
		and calculates the optimal power flow by \eqref{power}.
		\item Regardless of whether the switching switch is closed or not, the supply-demand balance needs to be met. Specifically, since the output power of each BESS is determined by its MC and upper and lower power limits, it requires us to design a suitable routing algorithm to ensure supply-demand balance, i.e.,
		$$\sum_{i=1}^{n}(D_i+P_{L,i}-P_i-gP_{UG})=0,$$
		where $D_i$ is the load at bus $i$.
		\item When designing MC consensus and routing schemes, UG plug and play should be met.
	\end{itemize}
	\subsection{Graph Theory}
	\par Assuming there are $n$ BESSs in a physical system, correspondingly, in an information system, there must be $n$ agents corresponding to them one by one, who communicate with each other over a communication topology graph ${\cal G}({\cal V},{\cal E})$. Therein, ${\cal V}=\{v_1,v_2,\cdots,v_n\}$ and ${\cal E}$ are vertex and edge sets, respectively. Agent $j$ is viewed as a neighbor of agent $i$ if the information flows from agent $j$ to $i$, which is denoted by $(v_i,v_j)\in{\cal E}$. 
	\par Some useful matrices are defined here. ${\cal A}=[a_{ij}]_{n\times n}$ is an adjacency matrix, where $a_{ij}=1$ if $(v_i,v_j)\in{\cal E}$ and $a_{ij}=0$ otherwise. ${\cal D}=diag(d_1,d_2,\cdots,d_n)$ is a degree matrix, where $d_i=\sum_{j=1}^{n}a_{ij}$ is viewed as the in-degree of agent $i$. Thus, a Laplacian matrix corresponding to $\cal G$ is defined as $L={\cal D}-{\cal A}$. For a leader agent, a matrix $B=diag(b_{11},b_{11},\cdots,b_{nn})$ is defined, where $b_{ii}=1$ if agent $i$ can access the leader, and $b_{ii}=0$ otherwise.
	\par Next, some definitions and properties regarding the matrices defined above are presented.
	\begin{assumption}
		\label{assu1}
		In this article, the communication topology used is represented as an undirected connected graph, which implies that $(v_j,v_i)\in{\cal E}$ if $(v_i,v_j)\in{\cal E}$ for $i,j\in\{1,2,\cdots,n\}$.
	\end{assumption}
	\begin{assumption}
		\label{assu2}
		Assuming that only a small number of agents can access the agent 0 to manage ER.
	\end{assumption}
	\begin{property}
		\label{p}
		\cite{niLeaderfollowingConsensusMultiagent2010} Let $H=L+B$. Under Assumption \ref{assu1}, the matrices $L$ and $H$ have the following properties.
		\begin{enumerate}
			\item $L$ is a non negative matrix, and 0 is a simple eigenvalue of it. That is, $0=\eta_1(L)<\eta_2(L)\le\cdots\le\eta_n(L)$, where $\eta_i(L)$ is a eigenvalue of $L$. In addition, $1^TL=0$. 
			\item $H$ is  positive definite, which implies $0<\eta_1(H)\le\eta_2(H)\le\cdots\le\eta_n(H)$.
		\end{enumerate}
	\end{property}
	\begin{definition}
		\label{def1}
		In the grid connected mode, the neighbors of the agent 0 are called flexible agents.
	\end{definition}
	\par Here, the Laplace matrix $L$ is split for future use. $-{\cal B}$ is a matrix, where each column either has all components of 0 or equals to the one corresponding to a flexible agent in ${\cal A}$. Then, ${\cal L}=L-{\cal B}$.
	
	\section{Distributed PI+Reset schemes for discrete time economic dispatch of BESS network}\label{III}
	\par In this section, a distributed ED scheme with a reset mechanism is designed and analyzed. Firstly, the designed average power mismatch estimator, the MC consensus controller, and the routing scheme are presented in Subsection \ref{secA}. Subsequently, the performance of the average power mismatch estimator and the MC consensus controller are well analyzed in Subsections \ref{secB} and \ref{secC}. Finally, the routing scheme and supply-demand balance are analyzed in Subsection \ref{secD} based on the first two subsections.
	\subsection{A Distributed ED Scheme of A BESS Network Considering UG Switching}\label{secA}
	\par The distributed MC consensus scheme with a reset mechanism for each BESS is designed as follows
	\begin{equation}
		\label{PIc}
		\lambda_i^{k+1}=\lambda_i^k-h_1\xi_{\lambda,i}^k-h_2\nu_{\lambda,i}^k+\sigma^k{\Delta{\hat P}_{a,i}},
	\end{equation}
	where $\xi_{\lambda,i}^k=\sum\limits_{i=1}^{n}a_{ij}(\lambda_i^k-\lambda_j^k)+gb_{ii}(\lambda_i^k-\rho)$, the integral term is represented as 
	\begin{equation}
		\label{resl}
		\nu_{\lambda,i}^k=\left\{\begin{aligned}
			\nu_{\lambda,i}+\xi_{\lambda,i},\; if\; \xi_{\lambda,i}^{k-1}\xi_{\lambda,i}^k> 0\\
			\xi_i^k,\; if\; \xi_{\lambda,i}^{k-1}\xi_{\lambda,i}^k\le 0,
		\end{aligned}\right.
	\end{equation} 
	$h_1$ and $h_2$ are both positive gains, ${\Delta{\hat P}_{a,i}}$ is average power mismatch estimated by agent $i$, $\sigma^k$ is a positive time-varying coupling gain and $\lim\limits_{k\to+\infty}\sigma^k=0$. For a detailed structure of this PI controller with a reset mechanism, please refer to Fig. \ref{figpir}. It is obvious that $\nu_{\lambda,i}^k$ and $\xi_{\lambda,i}^k$ are of the same sign.
	\begin{figure}
		\centering
		\includegraphics[width=7cm]{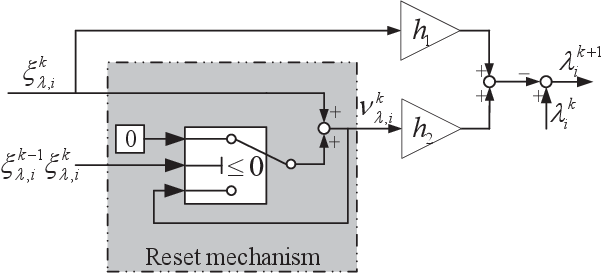} \caption{A PI controller with a error-dependent reset mechanism.\label{figpir}}
	\end{figure}
	\par A distributed MC consensus protocol is redesigned as follows by incorporating $\sigma^k{\Delta{\hat P}_{a,i}}$ into the closed-loop error,
	\begin{equation}
		\label{PIc1}
		\lambda_i^{k+1}=\lambda_i^k-h_1\xi_{\lambda,i}^k-h_2\nu_{\lambda,i}^k,
	\end{equation}
	where $\xi_{\lambda,i}^k=\sum\limits_{i=1}^{n}a_{ij}(\lambda_i^k-\lambda_j^k)+gb_{ii}(\lambda_i^k-\rho)-\sigma^k{\Delta{\hat P}_{a,i}^k}$.
	\par The compact form of the designed average power mismatch estimator are
	\begin{subequations}
		\label{mps0}
		\begin{equation}
			\label{mps01}	
			\Delta\hat P_a^{k+1}=\left\{\begin{aligned}
				&0,\,\mathrm{if}\,g=1\,\mathrm{and}\,\mathrm{agent}\,i\,\mathrm{is}\, \mathrm{flexible},\\
				&\Delta\hat P_a^k-z_1\xi^k_P-z_2\nu^k_P\\
				&+(1-gB)(\Delta P^{k+1}-\Delta P^k)\\
				&+(1-g)P_M^k,\,\mathrm{others},
			\end{aligned}\right.
		\end{equation}
		\begin{equation}
			\label{rems1}
			\nu_P^{k+1}=\left\{\begin{aligned}
				\nu_P^k+\xi_P^k,\; if\; \xi_{P,i}^{k+1}\xi_{P,i}^k> 0,\\
				\xi_P^k,\; if\; \xi_{P,i}^{k+1}\xi_{P,i}^k\le 0,
			\end{aligned}\right.
		\end{equation}
		\begin{equation}
			\label{PM1}
			P_M^{k+1}=g\left\{\begin{aligned}
				&P_M^k-z_1{\cal B}^T\Delta \hat P^k_a-z_2\sum_{j=l_i}^{k}{\cal B}^T\Delta \hat P^j_a\\
				&+B(\Delta P^{k+1}-\Delta P^k),\\
				&if\;k\in{\cal F}_P\;\mathrm{and}\;l_i\le k\le l_{i+1},\\
				&P_M^k-z_1{\cal B}^T\Delta \hat P^k_a
				+B(\Delta P^{k+1}-\Delta P^k),\\
				&if\;k\in{\cal J}_P\;\mathrm{and}\;k=l_i,
			\end{aligned}\right.
		\end{equation}
	\end{subequations}
	where $\Delta\hat P_a^k$ is a stack vector of the estimated power mismatch, $\Delta {\hat P}_a^0=(I-gB)\Delta P^0$, $\xi^k_P=(L-g{\cal B})\Delta\hat P_a^k$, $P_M^0=gB\Delta P^0$. From the definition of matrix $\cal B$, it can be inferred that if $b_{ii}=0$, $P_{M,i}^k=0$ for $k\in\{1,2,\cdots\}$ always holds.
	\par Here, a routing algorithm is given as follow
	\begin{equation}
		\label{PMGg}
		P_{UG}^{k+1}=1^T_nP_M^{k+1},
	\end{equation}
	where $P_{UG}^k$ denotes the power exchange between UG and BESSs.
	\begin{remark}
		It should be noted that the switching signal $g$ of a UG is not utilized globally. In \eqref{mps0} and \eqref{PIc}, it can be observed by the terms $gB$ and $(1-g)P_M^k$ that $g$ is provided by ER and only obtained by flexible agents. The scheme designed is effective in addressing the ED problem of a grid-connected or isolated BESS network. At the moment of UG switching, the corresponding changes only occur on the flexible agents.
	\end{remark}
	\subsection{Analysis of Distributed Estimation Scheme of The Average Power Mismatch}\label{secB}
	\par When the switching switch is closed, i.e., $g=1$, \eqref{mps01} is equivalent to
	\begin{equation}
		\label{mps}	
		\Delta\hat P_a^{k+1}=\Delta\hat P_a^k-z_1\xi^k_P-z_2\nu^k_P\\
		+(I-B)(\Delta P^{k+1}-\Delta P^k),
	\end{equation}
	where $\xi^k_P={\cal L}\Delta\hat P_a^k$, $\Delta {\hat P}_a^0=(I-B)\Delta P^0$, $\Delta\hat P_{a,i}^0=0$ if $b_{ii}=1$, $P_M^0=B\Delta P^0$. For the estimated average power mismatch, there is a lemma as follow summarized to illustrate consensus and stability.
	\begin{lemma}
		\label{DPa1}
		Under the Assumption \ref{assu1}, \eqref{mps}, and \eqref{rems1}, the estimated average power mismatch is bounded if $\frac{4z_2}{z_1^2}\le\eta_m({\cal L})$ with $z_1,z_2>0$, or $\frac{4z_2}{z_1^2}>\eta_m({\cal L})$ with $0<z_2<z_1$, which further converges to 0 if MC is stable.
	\end{lemma}
	\begin{proof}
		Define $y^k_P=col(\xi^P,\nu^P)$. Thus,
		$$\begin{aligned}
			y^k_P=\Phi^k({\cal L})y^0_P
			+\sum_{j=0}^{k-1}\Phi^{k-1-j}({\cal L})\left[\begin{matrix}
				(I-B)(\Delta P^{k+1}-\Delta P^k)\\
				0
			\end{matrix}\right],
		\end{aligned}$$
		where $\Phi({\cal L})=\left[\begin{matrix}
			I-h_1{\cal L} & -h_2{\cal L}\\
			I & I
		\end{matrix}\right]$. Under $\frac{4z_2}{z_1^2}\le\eta_m({\cal L})$ with $z_1,z_2>0$, or $\frac{4z_2}{z_1^2}>\eta_m({\cal L})$ with $0<z_2<z_1$, let $\alpha<1$ and $\Delta{\bar P}^k$ be the upper bounds of $\Vert\Phi({\cal L})\Vert$ and $\Vert\Delta{\hat P}_a^{k+1}-\Delta{\hat P}_a^k\Vert$, respectively. Thus, one can get
		$$\begin{aligned}
			\Vert y^k_P\Vert\le&\alpha^{k-l_{P,1}}\Vert y^{l_{P,1}+}_P\Vert+\sum_{j=l_{P,1}}^{k-1}\alpha^{k-1-j}\Delta\bar{P}^{j}\\
			\le&\alpha^{k-l_{P,1}}\Vert y^{l_{P,1}}_P\Vert+\sum_{j=l_{P,1}}^{k-1}\alpha^{k-1-j}\Delta\bar{P}^{j}\\
			\le&\alpha^k\Vert y^0_P\Vert+\sum_{j=0}^{k-1}\alpha^{k-1-j}\Delta\bar{P}^{j}\le+\infty.
		\end{aligned}$$
		If MCs are stable, $\Delta{\bar P}^k=o(1)$ with $k\to+\infty$, thus leading to $\lim\limits_{k\to+infty}\Vert y^k_P\Vert=0$. So far, Lemma \ref{DPa1} is proven.
	\end{proof}
	\par When $g=0$, that is, the switching switch is disconnected, the estimated average power mismatch is subjected to
	\begin{equation}
		\label{mp01}	
		\Delta\hat P_a^{k+1}=\Delta\hat P_a^k-z_1\xi^k_P-z_2\nu^k_P
		+(\Delta P^{k+1}-\Delta P^k)+P_M^k.
	\end{equation}	
	where $\xi^k_P=L\Delta\hat P_a^k$. Here, a lemma under a bounded $\Vert P_M^k\Vert$ is given to illustrate some properties of $\Delta\hat P_a^k$. Due to its similarity to Lemma \ref{DPa1}, its proof is omitted.
	\begin{lemma}
		\label{DPa2}
		Under the Assumption \ref{assu1}, \eqref{mp01}, and \eqref{rems1}, the estimated average power mismatch is bounded if $\frac{4z_2}{z_1^2}\le\eta_m(L)$ with $z_1,z_2>0$, or $\frac{4z_2}{z_1^2}>\eta_m(L)$ with $0<z_2<z_1$, which further converges to 0 if MC is stable.
	\end{lemma}
	\par To summarize the above statements, a proposition is provided to summarize the properties of our designed average power mismatch estimator when considering UG switching.
	\begin{proposition}
		\label{pro1}
		Under the Assumption \ref{assu1} and \eqref{mps0}, the estimated average power mismatch is bounded if $\frac{4z_2}{z_1^2}\le\eta_m({\cal L})$ with $z_1,z_2>0$, or $\frac{4z_2}{z_1^2}>\eta_m(L)$ with $0<z_2<z_1$, which further converges to 0 through subsequent proof of stability for MC.
	\end{proposition}
	\subsection{Analysis of Distributed MC Control Schemes}\label{secC}
	\par Define a set $l_\lambda:\{k|\xi_{\lambda,i}^{k-1}\xi_{\lambda,i}^k\le 0, i\in\{1,2,\cdots,n\}\}$ of jump times for the MC controller. So, a jump set set and a flow set are defined as ${\cal J}_\lambda:\{k|k\in l_\lambda\}$ and ${\cal F}_\lambda:\{k|k\notin l_\lambda\}$, respectively. Let $\xi_\lambda^k=[\xi^k_{\lambda,1},\xi^k_{\lambda,2},\cdots,\xi^k_{\lambda,n}]^T$, $\nu^k_\lambda=[\nu^k_{\lambda,1},\nu^k_{\lambda,2},\cdots,\nu^k_{\lambda,n}]^T$, and $\lambda^k=[\lambda^k_1,\lambda^k_2,\cdots,\lambda^k_n]$. Thus, 
	\begin{subequations}
		\label{nu}
		\begin{equation}
			\xi^k_\lambda=(L+gB)(\lambda^k-\rho),
		\end{equation}
		\begin{equation}
			\nu^k_\lambda=\left\{\begin{aligned}
				\nu^{k-1}_\lambda+\xi^{k-1}_\lambda,\; if\; k\in{\cal F}_\lambda\\
				\nu^{(k-1)+}_\lambda+\xi^{k-1}_\lambda,\; if\; k\in{\cal J}_\lambda.
			\end{aligned}\right.
		\end{equation}
	\end{subequations} 
	The compact form of \eqref{PIc} is
	\begin{equation}
		\label{lam}
		\lambda^{k+1}=\lambda^k-h_1\xi^k_\lambda-h_2\nu^k_\lambda+\sigma^k{\Delta{\hat P}_a^k}.
	\end{equation}
	\par Considering a grid-connected BESS network, i.e., $g=1$, the base system \cite{banosResetControlSystems2011} of \eqref{lam} is represented as
	\begin{equation}
		\label{gs}
		x^{k+1}_\lambda=\Phi(H) x^k_\lambda
	\end{equation}
	where $x^k_\lambda=[(\xi^k_\lambda)^T\quad (\nu^k_\lambda)^T]^T$, $\xi^k_\lambda=L(\lambda^k-\rho)$, $H=L+B$. The stability conditions of the base system are analyzed in the following lemma.
	\begin{lemma}
		\label{eigl}
		Under the Assumption \ref{assu1}, the base system \eqref{gs} is asymptotically stable if $\frac{4h_2}{h_1^2}\le\eta_1(H)$ with $h_1,h_2>0$, or $\frac{4h_2}{h_1^2}>\eta_1(H)$ with $0<h_2<h_1$, which leads to the system \eqref{gs} being regular \cite{banosResetControlSystems2011}.
	\end{lemma}
	\begin{proof}
		By checking the eigenvalues of the matrix $\Phi(H)$, under $\frac{4h_2}{h_1^2}\le\eta_1(H)$ with $h_1,h_2>0$, or $\frac{4h_2}{h_1^2}>\eta_1(H)$ with $0<h_2<h_1$, $\Vert\Phi(H)\Vert\le1$. According to the definition in \cite{banosResetControlSystems2011}, the base system \eqref{gs} is regular.
	\end{proof}
	\par Given the above analysis, an analysis of the consensus and stability of MCs driven by \eqref{PIc} when the switching switch is closed is summarized in the following theorem.
	\begin{theorem}
		\label{TH1g}
		Under Assumption \ref{assu1} and Lemma \ref{eigl}, MCs driven by \eqref{PIc} and \eqref{resl} can be achieve consensus and converge to the electric price for the grid connected mode. 
	\end{theorem}
	\begin{proof}
		Define $e^k_\lambda=[(\xi^k)^T\quad (\nu^k_\lambda)^T]^T$. With \eqref{PIc} and \eqref{resl} in hand, one can get
		\begin{equation}
			\label{Bgs}
			e^{k+1}_\lambda=\Phi(H) e^k_\lambda+\begin{bmatrix}
				\sigma^kH{\Delta {\hat P}}^k_a\\
				0
			\end{bmatrix}.
		\end{equation}
		Furthermore, 
		$$e^{k+1}_\lambda=(\Phi(H))^{k-l_i} e^{l_i+}_\lambda+\sum_{j=l_i}^{k}(\Phi(H))^{k-j}\begin{bmatrix}
			\sigma^jH{\Delta {\hat P}}^j_a\\
			0
		\end{bmatrix}.$$
		Taking the 2-norm of the above equation yields
		\begin{equation}
			\label{el1}
			\begin{aligned}
				&\Vert e^{k+1}_\lambda\Vert\le\Vert(\Phi(H))^{k-l_i} e^{l_i}_\lambda\Vert
				+\Vert\sum_{j=l_i}^{k}(\Phi(H))^{k-j}\begin{bmatrix}
					\sigma^jH{\Delta {\hat P}}^j_a\\
					0
				\end{bmatrix}\Vert\\
				\le&\Vert(\Phi(H))^k e^0_\lambda\Vert
				+\Vert\sum_{j=0}^{k}(\Phi(H))^{k-j}\begin{bmatrix}
					\sigma^jH{\Delta {\hat P}}^j_a\\
					0
				\end{bmatrix}\Vert.
			\end{aligned}
		\end{equation}
		According to Lemma \ref{eigl}, it can be easily obtained that
		$$\lim\limits_{k\to+\infty}\Vert(\Phi(H))^k e^0_\lambda\Vert=0.$$
		Given that $\Vert\Delta {\hat P}_a\Vert$ is bounded, there must be
		$$\begin{aligned}
			&\Vert\sum_{j=0}^{k}(\Phi(H))^{k-j}\begin{bmatrix}
				\sigma^jH{\Delta {\hat P}}^j_a\\
				0
			\end{bmatrix}\Vert\\
			\le&\sum_{j=0}^{k}\Vert(\Phi(H))^{k-j}\sigma^j\Vert\Vert\begin{bmatrix}
				H{\Delta {\hat P}}^j_a\\
				0
			\end{bmatrix}\Vert\\
			\le& o(1)\sum_{j=0}^{k}
			\Vert\begin{bmatrix}
				H{\Delta {\hat P}}^j_a\\
				0
			\end{bmatrix}\Vert,k\to+\infty.
		\end{aligned}$$
		Thus, $\lim\limits_{k\to+\infty}\Vert e^{k+1}_\lambda\Vert=0$, which leads to $\lim\limits_{k\to+\infty}\Vert (\lambda^k-\rho1)\Vert=0$ and $\lim\limits_{k\to+\infty}\Vert \nu^k_\lambda\Vert=0$. So far, Theorem \ref{TH1g} is proved.
	\end{proof}
	\par The compact form of \eqref{PIc1} is
	\begin{equation}
		\label{lam1}
		\lambda^{k+1}=\lambda^k-h_1\xi^k_\lambda-h_2\nu^k_\lambda,
	\end{equation}
	where $\xi^k_\lambda=H(\lambda^k-\rho1)-\sigma^k\Delta{\hat P}_a^k$, and whose base system still conforms to Lemma \ref{eigl}. Thus, the stability of the improved MC consensus controller \eqref{lam1} is analyzed in the following theorem.
	\begin{theorem}
		\label{TH2g}
		Under Assumption \ref{assu1} and Lemma \ref{eigl}, MCs driven by \eqref{lam1} and \eqref{resl} can be achieve consensus and converge to the electric price for the grid connected mode.
	\end{theorem}
	\begin{proof}
		According to \eqref{lam1}, one can get
		$$\xi^{k+1}_\lambda=\xi_\lambda^k-h_1H\xi_\lambda^k-h_1H\nu_\lambda^k+\sigma^{k+1}\Delta{\hat P}_a^{k+1}-\sigma^k\Delta{\hat P}_a.$$
		Besides,
		\begin{equation}
			\label{Bgs1}
			e^{k+1}_\lambda=\Phi(H) e^k_\lambda+\begin{bmatrix}
				\sigma^{k+1}H{\Delta {\hat P}}^{k+1}_a-\sigma^kH{\Delta {\hat P}}^k_a\\
				0
			\end{bmatrix}.
		\end{equation}
		Thus, it can be obtained by taking the 2-norm of the above equation that
		\begin{equation}
			\label{el2}
			\begin{aligned}
				\Vert e^{k+1}_\lambda\Vert\le&\Vert(\Phi(H))^{k-l_i} e^{l_i}_\lambda\Vert\\
				&+\Vert\sum_{j=l_i}^{k}(\Phi(H))^{k-j}\begin{bmatrix}
					\sigma^{j+1}H{\Delta {\hat P}}^{j+1}_a-\sigma^jH{\Delta {\hat P}}^j_a\\
					0
				\end{bmatrix}\Vert\\
				\le&\Vert(\Phi(H))^k e^0_\lambda\Vert\\
				&+\Vert\sum_{j=0}^{k}(\Phi(H))^{k-j}\begin{bmatrix}
					\sigma^{j+1}H{\Delta {\hat P}}^{j+1}_a-\sigma^jH{\Delta {\hat P}}^j_a\\
					0
				\end{bmatrix}\Vert,
			\end{aligned}
		\end{equation}
		where 
		$$\begin{aligned}
			&\Vert\sum_{j=0}^{k}(\Phi(H))^{k-j}\begin{bmatrix}
				\sigma^{j+1}H{\Delta {\hat P}}^{j+1}_a-\sigma^jH{\Delta {\hat P}}^j_a\\
				0
			\end{bmatrix}\Vert\\
			\le&\sum_{j=0}^{k}\Vert(\Phi(H))^{k-j}\sigma^j\Vert\Vert\begin{bmatrix}
				\frac{\sigma^{j+1}}{\sigma^j}H{\Delta {\hat P}}^{j+1}_a-H{\Delta {\hat P}}^j_a\\
				0
			\end{bmatrix}\Vert\\
			\le &o(1)\sum_{j=0}^{k}
			\Vert\begin{bmatrix}
				\frac{\sigma^{j+1}}{\sigma^j}H{\Delta {\hat P}}^{j+1}_a-H{\Delta {\hat P}}^j_a\\
				0
			\end{bmatrix}\Vert,k\to+\infty.
		\end{aligned}$$
		Thus, $\lim\limits_{k\to+\infty}\Vert e^{k+1}_\lambda\Vert=0$, which leads to Theorem \ref{TH2g}.
	\end{proof}
	\par When $g=0$, \eqref{lam} and \eqref{lam1} are reduced as
	\begin{subequations}
		\label{lem3}
		\begin{equation}
			\label{lem3a}
			\lambda^{k+1}=\lambda^k-h_1\xi^k_\lambda-h_2\nu^k_\lambda+\sigma^k{\Delta{\hat P}_a^k}
		\end{equation}
		\begin{equation}
			\label{lem3b}
			\lambda^{k+1}=\lambda^k-h_1\xi^k_\lambda-h_2\nu^k_\lambda,
		\end{equation}
	\end{subequations}
	where $\xi^k_\lambda=L\lambda^k$ and $\xi^k_\lambda=L\lambda^k-\sigma^k{\Delta{\hat P}_a^k}$, respectively. The proof of consensus follows a similar approach to Theorems \ref{TH1g} and \ref{TH2g}. Therefore, we only need to prove the stability of the consensus state.
	\begin{lemma}
		\label{lem4}
		Under Assumption \ref{assu1}, if $\frac{4h_2}{h_1^2}\le\eta_1(L)$ with $h_1,h_2>0$, or $\frac{4h_2}{h_1^2}>\eta_1(L)$ with $0<h_2<h_1$, MCs driven by \eqref{lem3} can asymptotically stabilize.
	\end{lemma}
	\begin{proof}
		Multiplying \eqref{lem3a} left by $1^T$ yields
		$$\lim\limits_{k\to+\infty}1^T(\lambda^{k+1}-\lambda^k)=\lim\limits_{k\to+\infty}\sigma^k1^T{\Delta{\hat P}_a^k}=0,$$
		which is ensured by $\lim\limits_{k\to+\infty}\sigma^k=0$ and the boundedness of $\Vert{\Delta{\hat P}_a^k}\Vert$. Similarly, multiplying \eqref{lem3b} left by $1^T$ yields
		$$\begin{aligned}
			\lim\limits_{k\to+\infty}1^T(\lambda^{k+1}-\lambda^k)=&\lim\limits_{k\to+\infty}\sigma^k1^T{\Delta{\hat P}_a^k}\\
			&+\lim\limits_{k\to+\infty}\sum_{t=l_i^P}^{k}\sigma^t1^T{\Delta{\hat P}_a^t}=0.
		\end{aligned}$$
		Thus, Lemma \ref{lem4} has been proven.
	\end{proof}
	\par At this point, we have the following proposition regarding the performance of the currently designed MC consensus solution in addressing the ED problem of a BESS network considering the switching of UG.
	\begin{proposition}
		\label{prop}
		Regardless of choosing control scheme \eqref{PIc} or \eqref{PIc1}, according to Theorems \ref{TH1g} and \ref{TH2g}, MC can achieve consensus and converge to the electricity price when $g=1$.
	\end{proposition}
	\begin{remark}
		\label{re4}
		Here, the progressiveness of the two MC control schemes \eqref{PIc} and \eqref{PIc1} is explained. Due to the nonlinear relationship between $\Delta P_i^k$ and $\lambda_i^k$, it is difficult to accurately describe the convergence rate of MCs. But there are still clues. \eqref{el1} and \eqref{el2} respectively provide upper bounds on the control error for two schemes. Given the boundedness of ${\Delta {\hat P}}_{a,i}^k$, it is assumed that there exists a positive constant ${\Delta {\bar P}}$ that ensures that $\vert{\Delta {\hat P}}_{a,i}^k\vert\le{\Delta {\bar P}}$ always holds. Therefore, $\Vert\sigma^k{\Delta {\hat P}}_a^k\Vert\le\sigma^k{\Delta {\bar P}}$ and $\Vert\sigma^{k+1}{\Delta {\hat P}}_a^{k+1}-\sigma^k{\Delta {\hat P}}_a^k\Vert\le(\sigma^k-\sigma^{k+1}){\Delta {\bar P}}$, indicating that the upper bound of the control error for MCs driven by \eqref{PIc1} is always smaller than that driven by \eqref{PIc}. In other words, applying \eqref{PIc1} promotes faster convergence of MCs to the electricity prices. In addition, the rapid convergence of MC leads to $\Delta P_i^k$ converging to 0 faster, thereby accelerating the convergence rate of $\Delta {\hat P}_{a,i}^k$. In summary, choosing \eqref{PIc1} as the control scheme for MC is currently a wise choice.
	\end{remark}
	\subsection{Analysis of Supply and Demand Balance and Utility Grid Switching}\label{secD}
	\par From the above analysis, it can be seen that regardless of the mode of a BESS network in the previous time, MCs is bounded at the switching time and can eventually stabilize and reach consensus afterwards. Thus, a corollary on the stability of the average power mismatch estimator is drown as follow. 
	\begin{corollary}
		\label{cor}
		According to Propositions \ref{pro1} and \ref{prop}, the estimated average power mismatch can eventually converge to 0 whether $g=1$ or 0.
	\end{corollary}
	\par In addition, there is a corollary about the properties of $P_M^k$ as follows.
	\begin{corollary}
		\label{coro1}
		When $g=1$, According to Lemma \ref{DPa1}, it can be inferred that $P_{M,i}^k$ for $i\in\{1,2,\cdots,n\}$ is asymptotically stable.
	\end{corollary}
	\begin{proof}
		According Definition \ref{def1} and \eqref{mps}, it can be easily obtained that
		$$\begin{bmatrix}
			P_M^{k+1}\\
			\Delta{\hat P}_a^{k+1}
		\end{bmatrix}=\begin{bmatrix}
			P_M^k\\
			\Delta{\hat P}_a^k
		\end{bmatrix}-z_1\xi^k-z_2{\nu}^k+\Delta{\cal P}^k,$$
		where $\xi^k={\cal H}\begin{bmatrix}
			P_M^k\\
			\Delta{\hat P}_a^k
		\end{bmatrix}=\begin{bmatrix}
			0&{\cal B}^T\\
			0&{\cal L}
		\end{bmatrix}\begin{bmatrix}
			P_M^k\\
			\Delta{\hat P}_a^k
		\end{bmatrix}$, $\Delta{\cal P}^k=\begin{bmatrix}
			B(\Delta P^{k+1}-\Delta P^k)&0\\
			0&(I-B)(\Delta P^{k+1}-\Delta P^k)\\
		\end{bmatrix}$, $\nu^k=\left\{\begin{aligned}
			\nu^{k-1}+\xi^{k-1},\,\mathrm{if}\,k\in{\cal F}_P,\\
			\xi^{k-1},\,\mathrm{if}\,k\in{\cal J}_P.
		\end{aligned}\right.$ Thus,
		$$\begin{bmatrix}
			\xi^{k+1}\\
			\nu^{k+1}
		\end{bmatrix}=\Phi({\cal H})\begin{bmatrix}
			\xi^k\\
			\nu^k
		\end{bmatrix}+\begin{bmatrix}
			\Delta{\cal P}&0\\
			0&0
		\end{bmatrix}.$$
		According to Lemma \ref{DPa1} and Proposition \ref{prop}, $\Vert\begin{bmatrix}
			\xi^{k+1}\\
			\nu^{k+1}
		\end{bmatrix}\Vert\to 0$ with $k\to +\infty$,
		which leads to $\lim\limits_{k\to+\infty}\Vert\begin{bmatrix}
			P_M^{k+1}\\
			\Delta{\hat P}_a^{k+1}
		\end{bmatrix}-\begin{bmatrix}
			P_M^k\\
			\Delta{\hat P}_a^k
		\end{bmatrix}\Vert=0$.
	\end{proof}
	\par To ensure supply-demand balance, certain conditions need to be met at the time when the UG is switched. Thus, two lemmas are given to illustrate the power supply-demand balance in the grid-connected and isolated modes.
	\begin{lemma}
		\label{leg}
		Due to $L={\cal L}+{\cal B}$, $1^TL=0_n$ and $1^T[\Delta {\hat P}^0_a-\Delta P^0]+P_{UG}^0=0$, the power supply-demand balance of a grid-connected BESS network, i.e., $g=1$, can be asymptotically guaranteed by \eqref{mps} with \eqref{rems1}, \eqref{lam} or \eqref{lam1} with \eqref{nu}, and \eqref{PMGg}.
	\end{lemma}
	\begin{proof}
		With the premises of $L={\cal L}+{\cal B}$, $1^TL=0_n$ and $1^T[\Delta {\hat P}^0_a-\Delta P^0]+P_{UG}^0=0$, there must be
		$$\begin{aligned}
			&1^T[\Delta {\hat P}^{k+1}_a-\Delta P^{k+1}]+P_{UG}^{k+1}\\
			=&1^T[(I-z_1{\cal L})\Delta {\hat P}^k_a-z_2v_P^k-\Delta P^k]+P_{UG}^k
			-z_11^T_n{\cal B}^T\Delta \hat P^k_a\\
			&-z_2\sum_{j=l_{P,i}}^{k}1^T_n{\cal B}^T\Delta \hat P^j_a\\
			=&1^T[(I-z_1L)\Delta {\hat P}^k_a-z_2L\sum_{j=l_{P,i}+1}^{k}\Delta {\hat P}^k_a]
			+P_{UG}^k\\
			=&1^T[\Delta {\hat P}^k_a-\Delta P^k]+P_{UG}^k
			=\cdots\\
			=&1^T[\Delta {\hat P}^{l_{P,i}+1}_a-\Delta P^{l_{P,i}+1}]+P_{UG}^{l_{P,i}+1}\\
			=&1^T[(I-z_1L)\Delta {\hat P}^{l_{P,i}}_a-\Delta P^{l_{P,i}}]+P_{UG}^{l_{P,i}}\\
			=&1^T[\Delta {\hat P}^{l_{P,i}}_a-\Delta P^{l_{P,i}}]+P_{UG}^{l_{P,i}}\\
			=&\cdots
			=1^T[\Delta {\hat P}^0_a-\Delta P^0]+P_{UG}^0=0,
		\end{aligned}$$	
		which are ensured by \eqref{mps} with \eqref{rems1}, \eqref{lam} or \eqref{lam1} with \eqref{nu}, and \eqref{PMGg}.
	\end{proof}
	\begin{lemma}
		\label{lei}
		Due to $L={\cal L}+{\cal B}$, $1^TL=0_n$ and $1^T{\Delta {\hat P}^0_a-\Delta P^0}=0$, the power supply-demand balance of an isolated BESS network, i.e., $g=0$, can be asymptotically guaranteed by \eqref{mps01} with \eqref{rems1}, and \eqref{lam} or \eqref{lam1} with \eqref{nu}.
	\end{lemma}
	\par For the proof of Lemma \ref{lei}, it is similar to the one of Lemma \ref{leg} and therefore omitted. Based on the above analysis, a proposition is presented to illustrate the supply-demand balance of a BESS network with a switching UG.
	\begin{proposition}
		For a BESS network with a switching UG, the output power of BESSs driven by MCs control schemes \eqref{lam} or \eqref{lam1} with \eqref{nu} and power limitations and the routing scheme \eqref{PMGg} can be derived that the power supply and demand balance can be asymptotically guaranteed.
	\end{proposition}
	\begin{proof}
		According to Lemma \ref{leg}, $1^T[\Delta {\hat P}^k_a-\Delta P^k]+P_{UG}^k=0$ when $g=1$. When the switching switch is disconnected, i.e. $g$ goes from 1 to 0, according to Corollary \ref{coro1} and \eqref{mp01}, there must be
		$$\begin{aligned}
			1^T[\Delta {\hat P}^{k+1}_a-\Delta P^{k+1}]
			=1^T[\Delta {\hat P}^k_a-\Delta P^k]+P_{UG}^k
			=0,
		\end{aligned}$$
		which leads that power supply and demand balance is kept asymptotically according to Lemma \ref{lei}.
	\end{proof}
	\section{Some Simulation Cases}\label{IV}
	\par In order to comprehensively test the designed solution, five cases are arranged and corresponding simulations are performed.
	\par In Case 1, the designed schemes are proven to be effective by executing some simulations on a BESS network considering UG switching as shown in Fig. \ref{fig3}. Meanwhile, in order to illustrate the progressiveness, the simulation results are compared with those under the existing scheme \cite{chenDistributedEconomicDispatch2021}.
	\par According to Definition \ref{def1}, agent 1 is flexible. So, according to the accompanying communication network, matrices required for simulation are
	$$L=\left[\begin{matrix}
		3&-1&-1&-1\\
		-1&2&0&-1\\
		-1&0&2&-1\\
		-1&-1&-1&3
	\end{matrix}\right],{\cal L}=\left[\begin{matrix}
		0&0&0&0\\
		-1&2&0&-1\\
		-1&0&2&-1\\
		-1&-1&-1&3
	\end{matrix}\right],$$ $$B=\left[\begin{matrix}
		1&0&0&0\\
		0&0&0&0\\
		0&0&0&0\\
		0&0&0&0
	\end{matrix}\right], {\cal B}=\left[\begin{matrix}
		3&0&0&0\\
		-1&0&0&0\\
		-1&0&0&0\\
		-1&0&0&0
	\end{matrix}\right].$$ 
	\par Subsequently, two cases, i.e., Cases 2 and 3, are arranged to further test the performance of the designed scheme under a time-phased electricity pricing environment and load switching.
	\par To test whether the designed solution has the plug and play capability of an agent/BESS, Case 4 is designed to investigate the performance of the designed solution in both scenarios.
	\begin{figure}
		\centering
		\includegraphics[width=7.5cm]{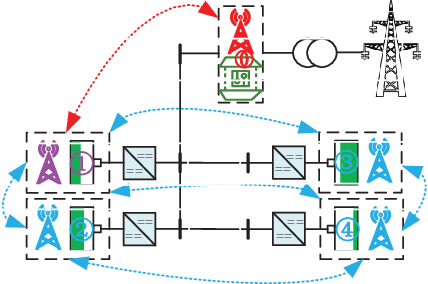} \caption{A BESS network with an ER and their communication graph.\label{fig3}}
	\end{figure}
	\subsection{Case 1. The Test on The Effectiveness of The Designed Scheme}
	\begin{figure}
		\centering
		\includegraphics[width=7.5cm]{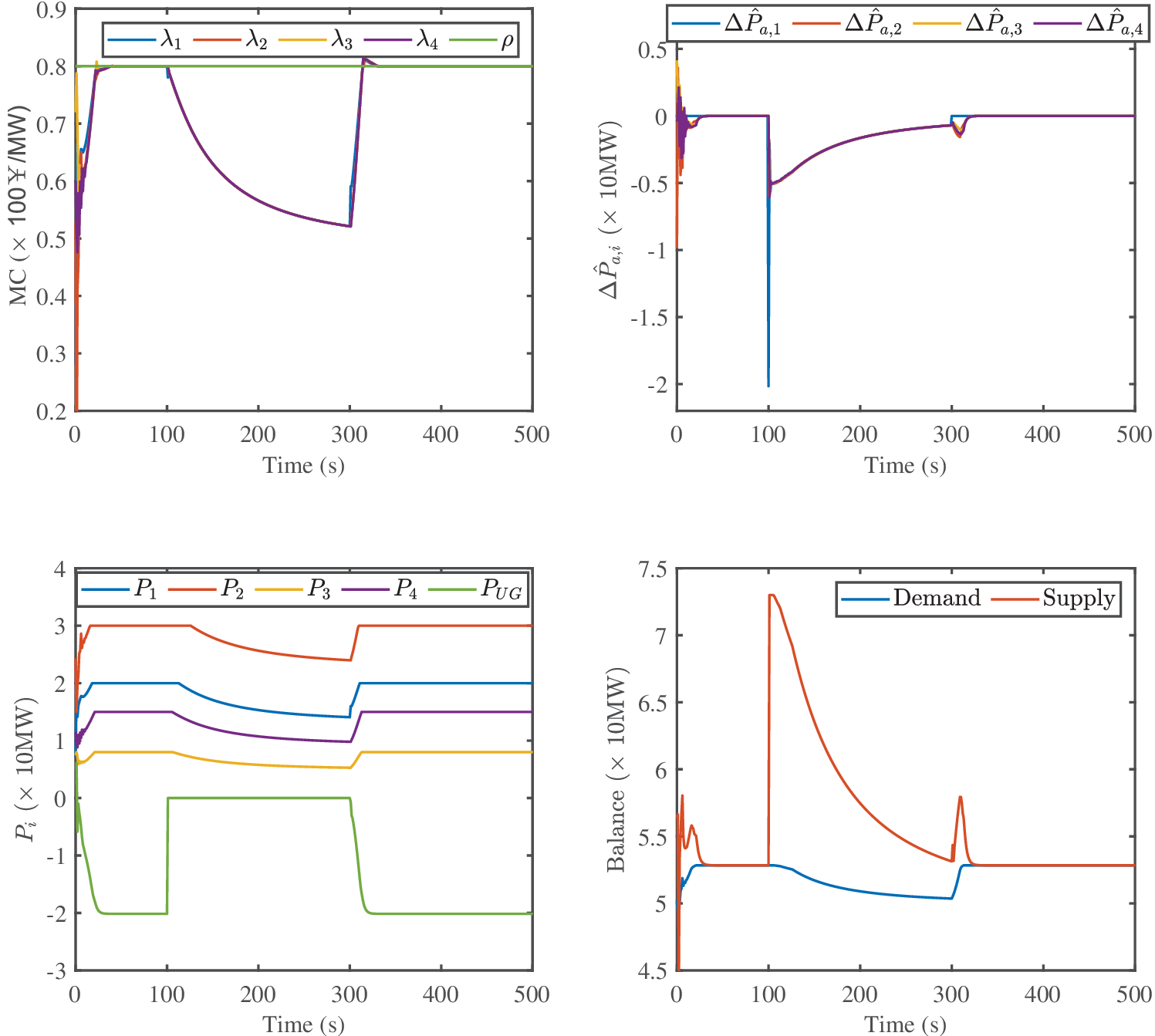} \caption{Simulation results of a distributed scheme with PI+Reset controllers in Case 1.\label{figc1}}
	\end{figure}
	\par The designed schemes and its comparison are tested in this case, where the scheme used for comparison comes from \cite{chenDistributedEconomicDispatch2021}. Starting from $k=0$, the switching switch is closed. When $k=100s$, the switching switch is disconnected. When $k=300s$, the switching switch is closed. The simulation results of the two designed schemes and the comparison can be found in Figs. \ref{figc1}-\ref{figc1c}.
	\begin{figure}
		\centering
		\includegraphics[width=7.5cm]{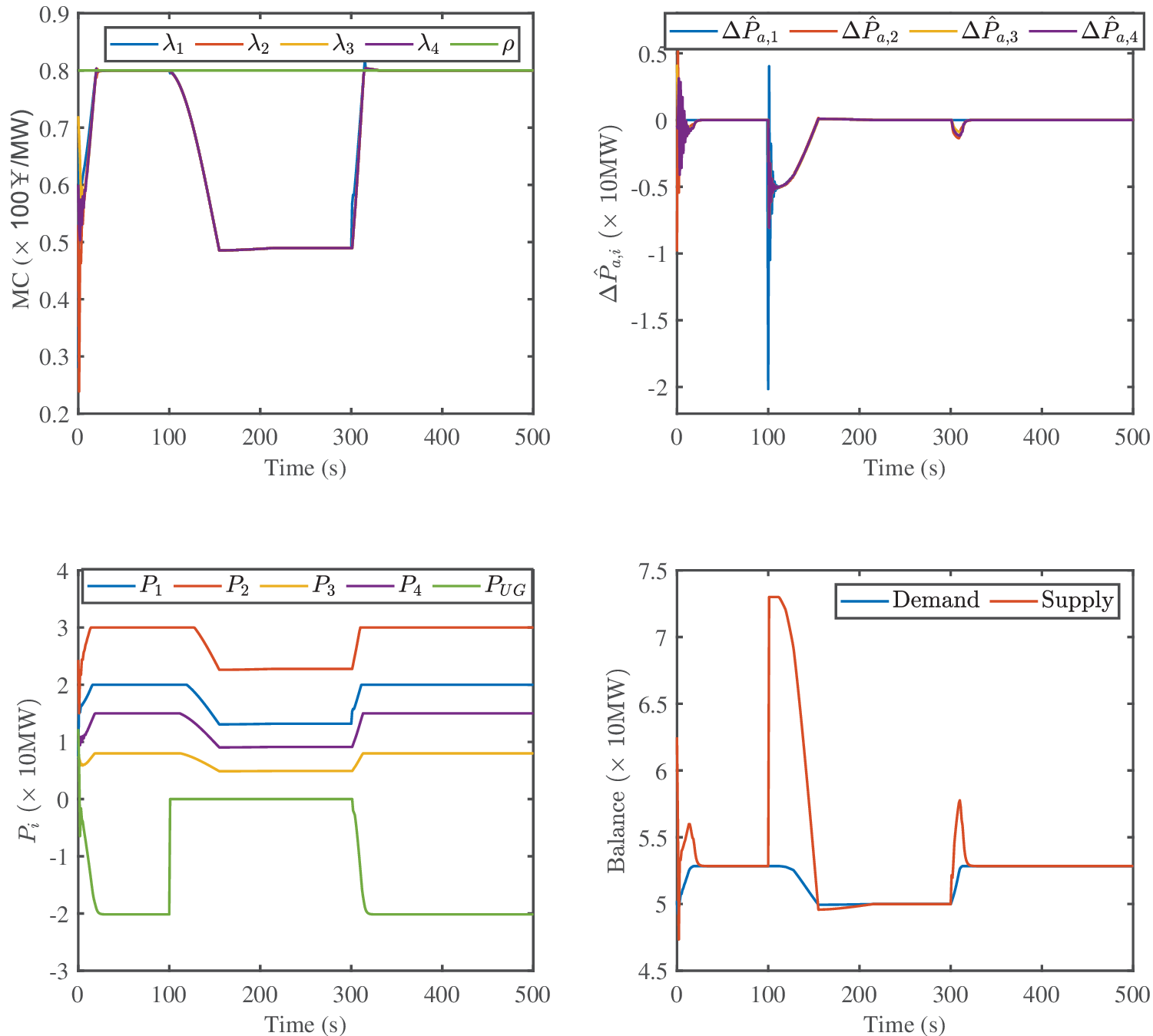} \caption{Simulation results of a distributed scheme \cite{chenDistributedEconomicDispatch2021} in Case 1.\label{figc1b}}
	\end{figure}
	\begin{figure}
		\centering
		\includegraphics[width=7.5cm]{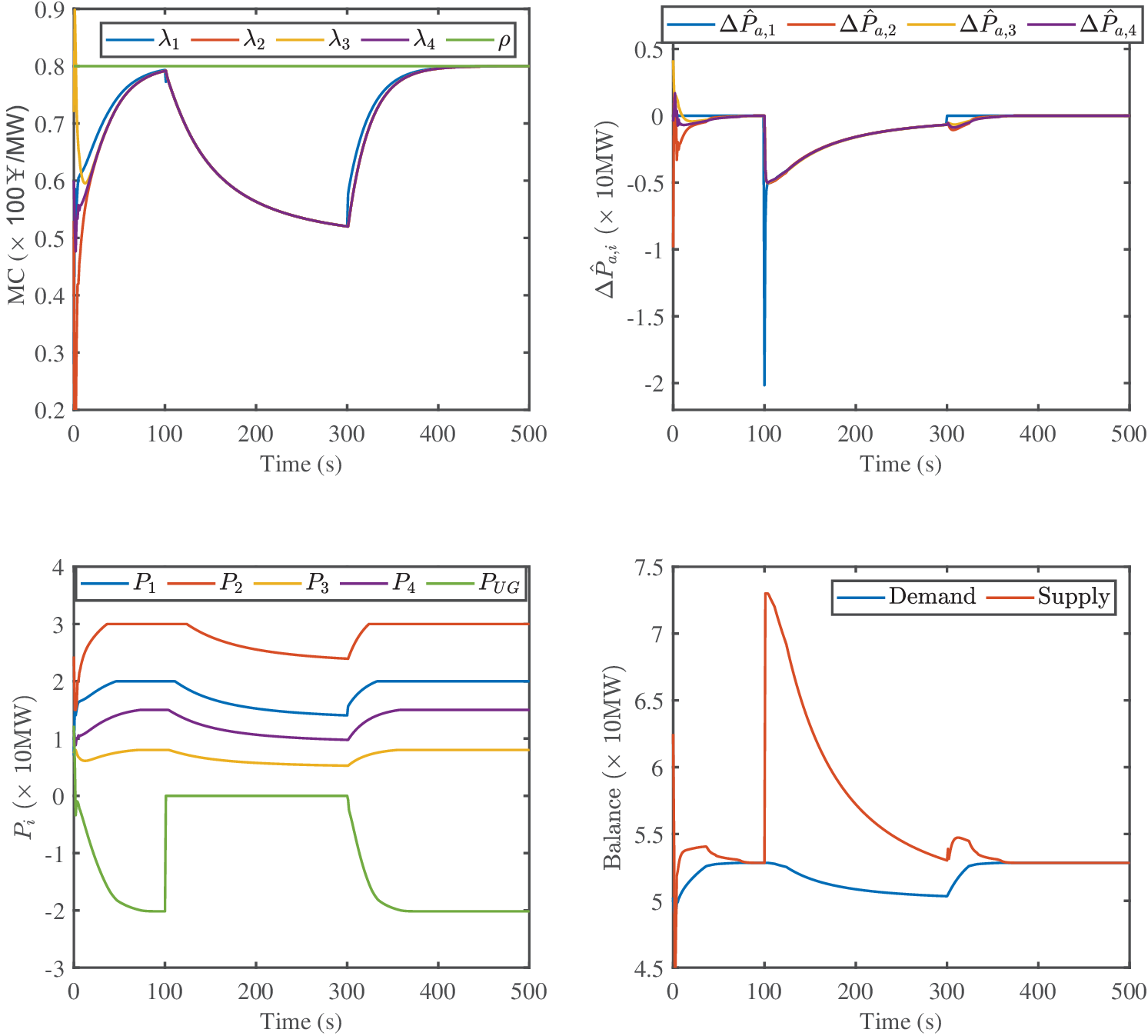} \caption{Simulation results of a distributed scheme \cite{chenDistributedEconomicDispatch2021} in Case 1.\label{figc1c}}
	\end{figure}
	\par From the simulation results in Fig. \ref{figc1}, it can be seen that under the MC consensus control scheme \eqref{PIc} and the average power mismatch estimation scheme \eqref{mps01}, the MCs and the average power mismatch estimated values can reach consensus, respectively. When the switching switch is closed, i.e. $g=1$, the MCs and the average power mismatch estimated values can converge to electricity price and 0, respectively, and the output of the UG and each BESS output is stable, ensuring a balance between power supply and demand. However, when the switching switch is disconnected, i.e. $g=0$, the convergence speed of the MCs and average power mismatch estimates slows down.
	\par By introducing $\sigma^k\Delta{\hat P}_{a,i}^k$ into the closed-loop error $\xi_{\lambda,i}^k$, driven by scheme \eqref{PIc1}, as shown in Fig. \ref{figc1c}, the dynamic performance of each variable does not change significantly in the grid-connected mode. However, the dynamic performance of variables in the isolated mode has been greatly improved.
	\par Comparing Fig. \ref{figc1} with Fig. \ref{figc1c}, the designed solutions \eqref{PIc} and \eqref{mps01} have certain advantages in solving the ED problem of a BESS network considering UG switching. Especially when the switching switch is closed, the MC and the average power mismatch estimated value converge faster to electricity price and 0, respectively. The advantage is further expanded by applying the MC control scheme \eqref{PIc1}, which greatly improves the dynamic performance in the isolated mode while ensuring that the dynamic performance of each variable is basically not affected in the grid-connected mode, as analyzed in Remark \ref{re4}. Anyway, comparing the three schemes, the designed schemes \eqref{PIc1} and \eqref{mps01} have better dynamic performance and control accuracy in dealing with the ED problem of a BESS network considering UG switching.
	\begin{figure}
		\centering
		\includegraphics[width=7.5cm]{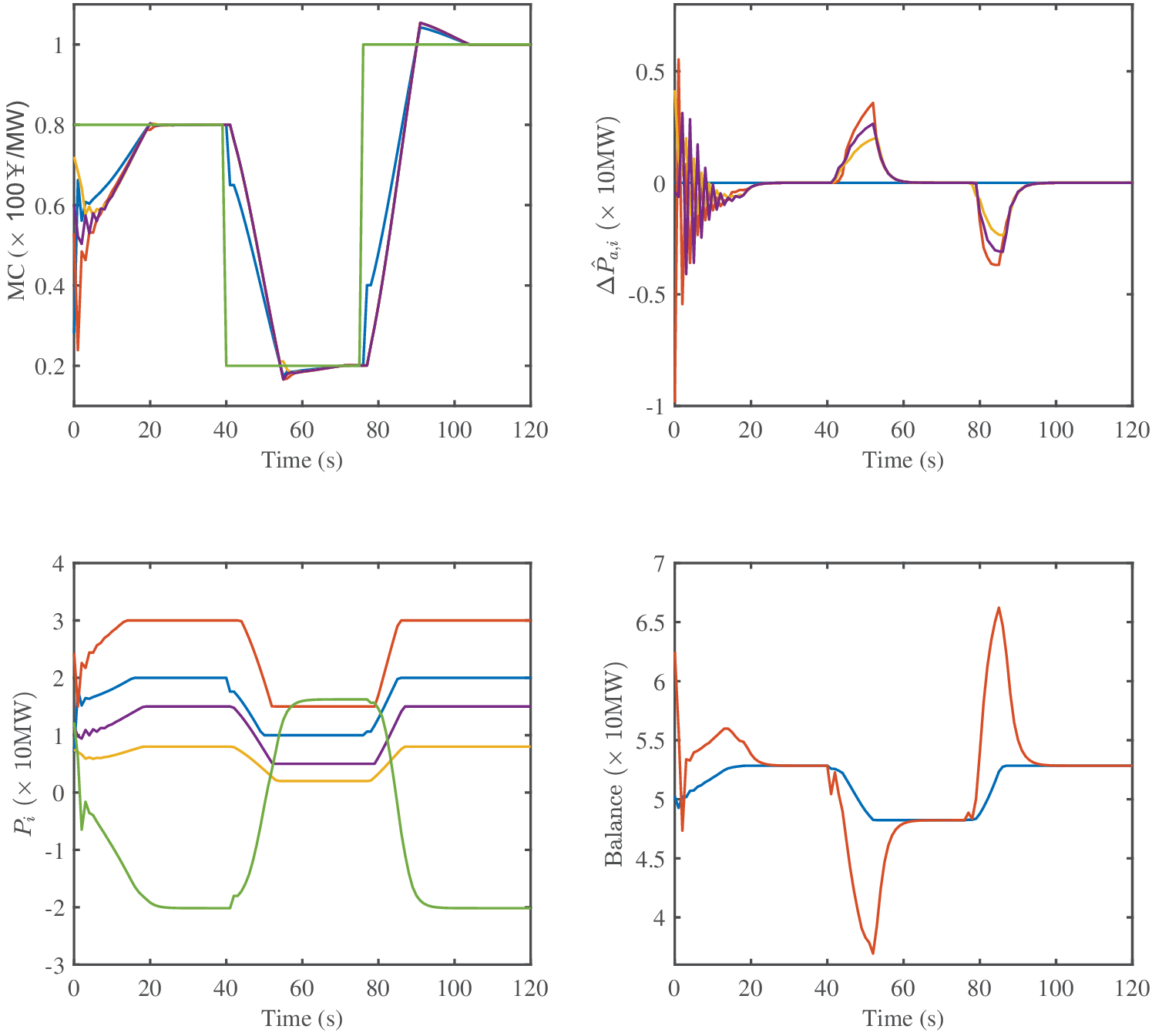} \caption{Simulation results of a distributed scheme with PI+Reset controllers under the time phased electricity price environment in Case 2.\label{figc3}}
	\end{figure}
	\subsection{Case 2. The Test in The Time Phased Electricity Price Environment}
	\par Time phased electricity prices are adopted by power companies to guide users to use electricity reasonably. The schemes \eqref{PIc1} and \eqref{mps01} are tested in a time phased electricity pricing environment, where $\rho=80 (0\le k<40s)$, $\rho=20 (40s\le k<80s)$, and $\rho=100 (80s\le k\le120s)$. The simulation results are shown in Fig. \ref{figc3}.
	\par From the simulation results, it can be seen that the designed scheme is sufficient to address the ED problem in a time phased electricity pricing environment. Specifically, at each time period, MC is able to converge to the electricity price, and the estimated average power mismatch converges to 0, which lead to stable output of UG and each BESS, as well as supply-demand balance of power. It is worth noting that near the moment of electricity price change, the average power mismatch estimated by each agent is affected to varying degrees. This is due to the introduction of feedback coupling term $\sigma^k{\Delta{\hat P}_a^k}$ into the MC control scheme, which is compared to the simulation results in \cite{chenDistributedEconomicDispatch2021}.
	\subsection{Case 3. The Test on Load Switching}
	\begin{figure}
		\centering
		\includegraphics[width=7.5cm]{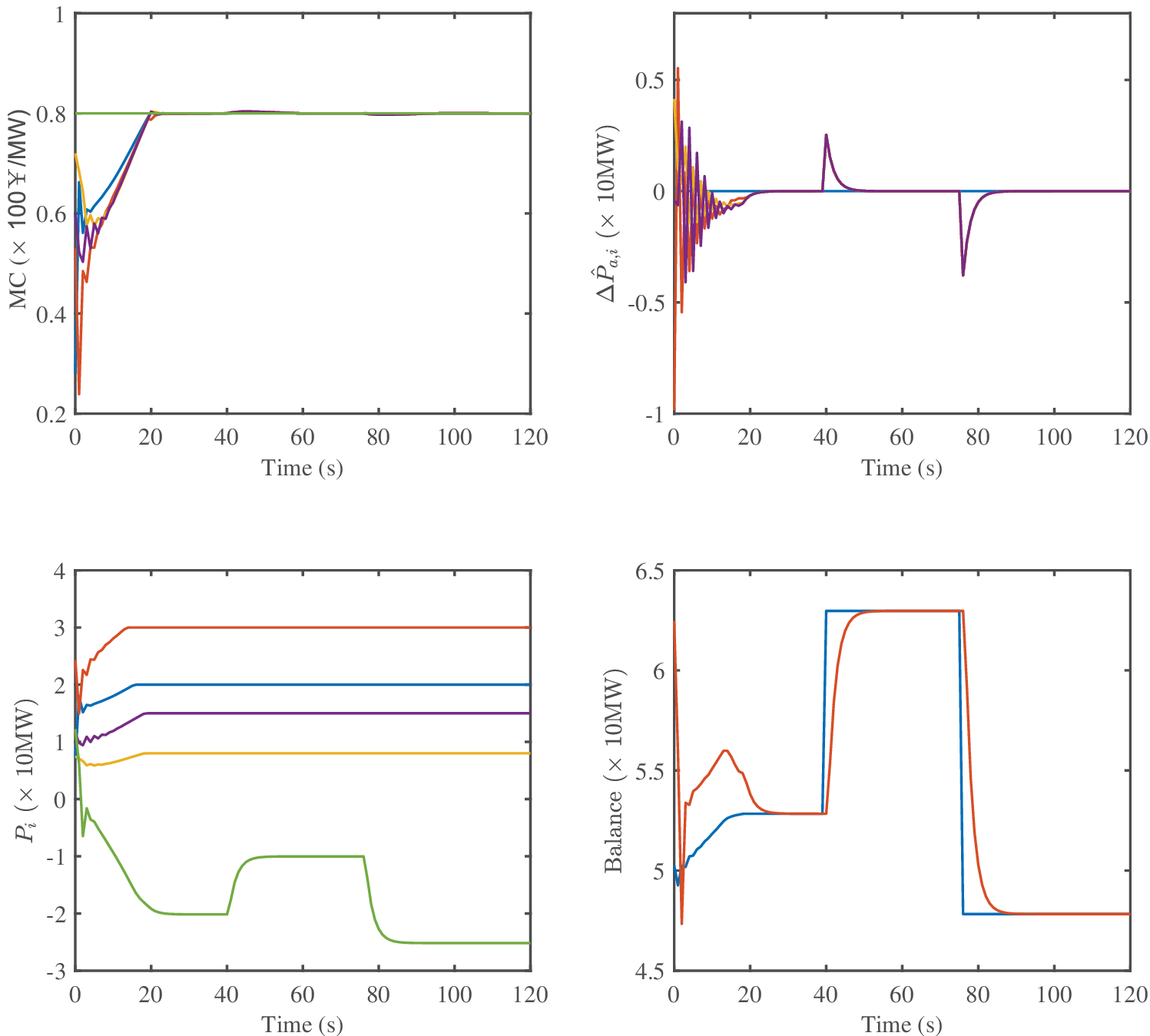} \caption{Simulation results on load switching in Case 3.\label{figc4}}
	\end{figure}
	\par This case is arranged to test the effectiveness of the designed solution in dealing with load switching. Therein, the load increases at $k=40s$ and decreases at $k=80s$. The simulation results are shown in Fig. \ref{figc4}.
	\par From the simulation results, it can be seen that, comparing to the simulation results in \cite{chenDistributedEconomicDispatch2021}, due to the introduction of the feedback coupling term $\sigma^k{\Delta{\hat P}_a^k}$, MCs are slightly affected during load switching but can achieve consensus, and the estimated average power mismatch can still converge to 0. Due to output limitations and electricity prices, the output power of each BESS remains constant, and the changes in load power are mainly borne by UG to ensure supply-demand balance.
	\begin{figure}
		\centering
		\includegraphics[width=7.5cm]{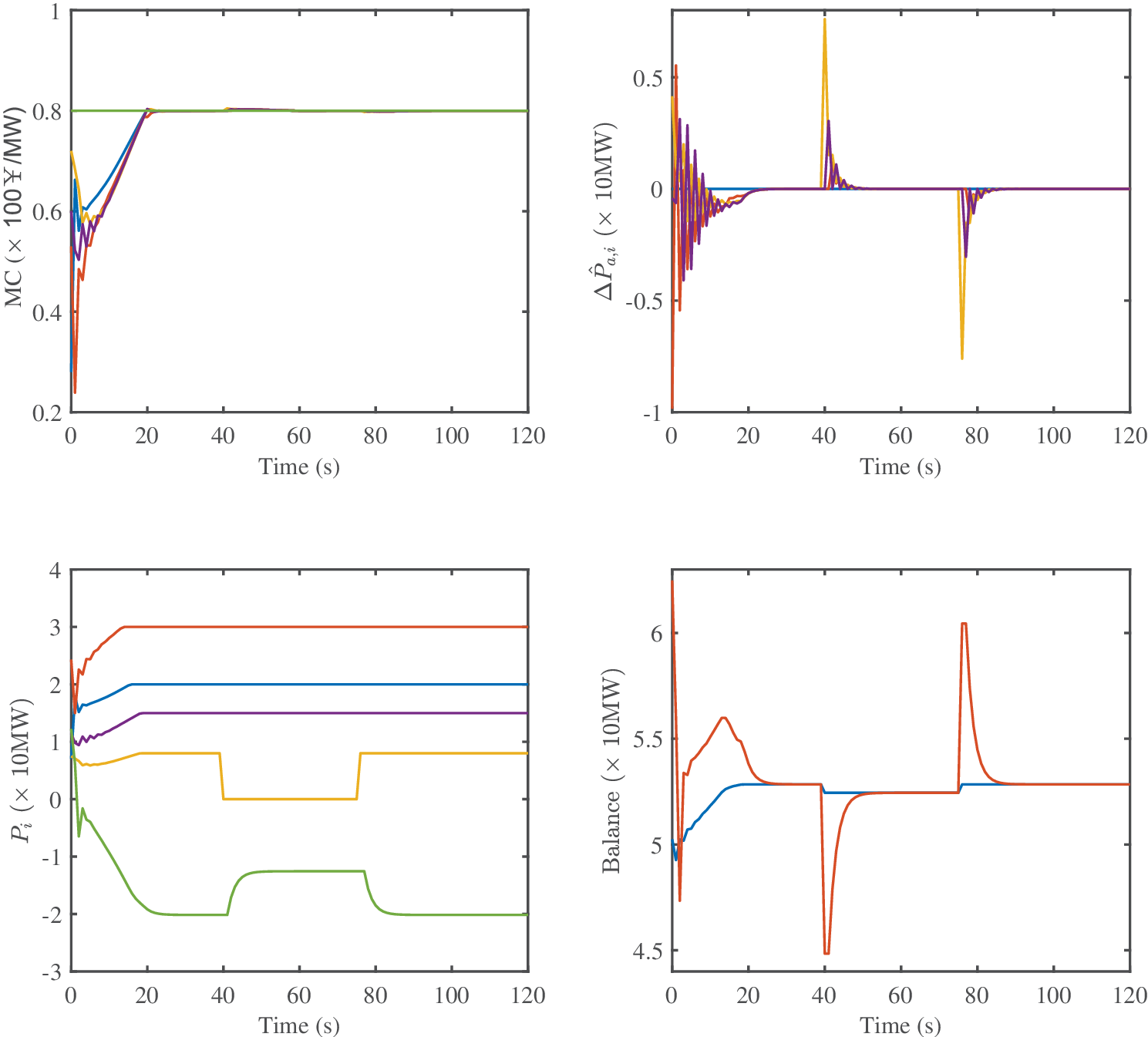} \caption{Simulation results on load switching in Case 4.\label{figc5}}
	\end{figure}
	\begin{figure}
		\centering
		\includegraphics[width=7.5cm]{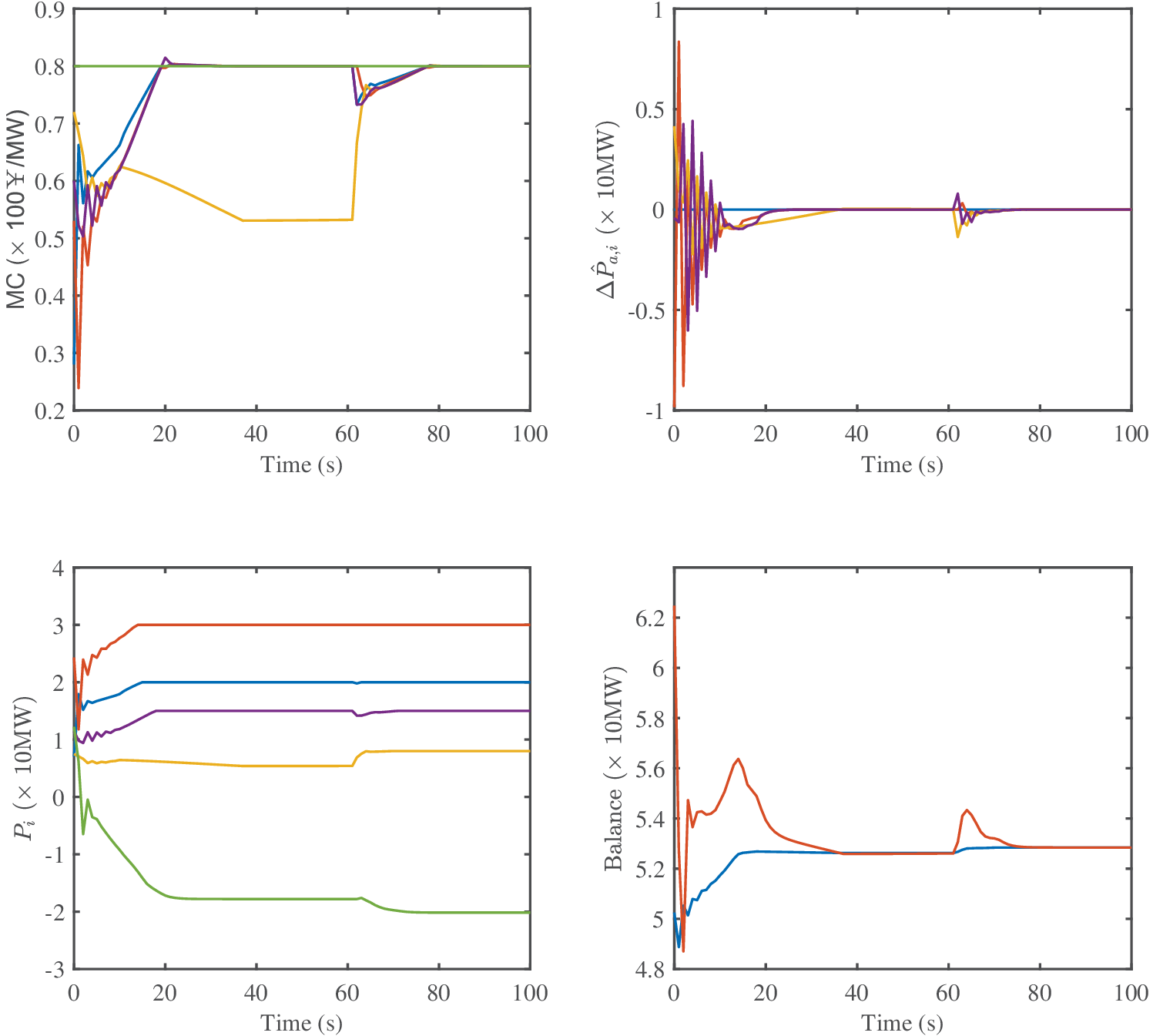} \caption{Simulation results on load switching in Case 4.\label{figc5b}}
	\end{figure}
	\subsection{Case 4. The Test on The Plug and Play}
	\par For physical systems, a BESS often exits the power supply sequence due to low level of SoC. For information systems, communication may be interrupted due to link attacks. So, this case is arranged to test the designed algorithm considering a play-and-plug BESS/agent.
	\par Firstly, a simulation considering plug-and-play BESS 3 is performed, where BESS 3 exits the power supply sequence when $k=40s$ and enters when $k=80s$. The simulation results are shown in Fig. \ref{figc5}. Secondly, a simulation considering plug-and-play agent 3 is executed, where agent 3 disconnects communication with other agents when $k=10s$ and resumes communication when $k=60s$. The simulation results are shown in Fig. \ref{figc5b}.
	\par From Fig. \ref{figc5}, it can be seen that although the exit and access of BESS 3 cause fluctuations in the MCs and the average power mismatch estimated values, MC can still converge to the electricity price, and the estimated value of average power mismatch can converge to 0. Based on this, the balance between power supply and demand can still be maintained, and the shortfall is supplemented by UG. From Fig. \ref{figc5b}, it can be seen that even if agent 3 does not communicate with other agents, as long as the remaining part satisfies Assumption \ref{assu1}, the MCs of other BESSs can converge to the electricity price, and the power supply and demand balance can still be guaranteed, and the power mismatch can ensure convergence to 0. From the above analysis, it can be seen that the designed solution has plug and play functionality, which is suitable for not only BESSs in physical systems, but also agents in information systems.
	\section{Conclusions}\label{V}
	\par The designed unified distributed solution can solve the ED problem of a BESS network considering UG switching. The introduction of a reset mechanism improves the convergence rate and control accuracy of each state of a BESS with almost no overshoot. The approach of introducing the feedback coupling term $\sigma^k{\Delta {\hat P}_a}$ into the control error of MC hardly affects the performance of the algorithm in the grid-connected mode, but greatly improves the convergence rate of each BESS state after the UG switching switch is disconnected. At the time of UG switching, there is no need to shut down to adapt to UG access and exit. In this sense, the designed solution has a certain plug and play feature in dealing with UG switching. The above conclusions are drawn through some simulation results and their comparison and discussion. However, some unresolved issues still exist. To ensure power supply and demand balance, a centralized reset mechanism is introduced into the average power mismatch estimator. In addition, the switching of UG requires resetting the flexible intelligent agents in the information system and decomposing the communication network. All of these are areas worth striving for in the future.

\bibliographystyle{IEEEtran}
\bibliography{re1.bib}

\end{document}